\documentclass[aps,pra,twocolumn,superscriptaddress,longbibliography,nofootinbib]{revtex4-2}
\usepackage[T1]{fontenc}
\usepackage[main=english,french]{babel}
\makeatletter
\@ifundefined{l@fr}{\@namedef{l@fr}{\l@french}}{}
\@ifundefined{captionsfr}{\@namedef{captionsfr}{\captionsfrench}}{}
\@ifundefined{datefr}{\@namedef{datefr}{\datefrench}}{}
\makeatother

\usepackage{etex}
\usepackage{braket}
\usepackage{physics}
\usepackage{amsmath,amssymb,amsthm,amstext,mathrsfs}
\usepackage[colorlinks=true,citecolor=blue,urlcolor=blue]{hyperref}
\usepackage{comment}
\usepackage[pdftex]{graphicx}
\usepackage{enumerate}
\usepackage{times,txfonts}
\usepackage{braket}
\usepackage{svg}
\usepackage{color}
\usepackage{natbib}
\usepackage{amsmath,blkarray}
\usepackage{mathtools}
\usepackage{latexsym}
\usepackage{tabularx, booktabs}
\usepackage{graphics,epstopdf}
\usepackage{graphicx, lipsum}

\usepackage{amsfonts,dsfont}
\usepackage{subcaption}
\usepackage{enumerate}
\usepackage{color,soul}

\newcommand{\be}{\begin{equation}}
\newcommand{\ee}{\end{equation}}
\newcommand{\ba}{\begin{eqnarray}}
\newcommand{\ea}{\end{eqnarray}}

\newcommand{\half}{\frac{1}{2}}

\newtheorem{thm}{Theorem}
\newtheorem{Lemma}{Lemma}

\definecolor{ss}{RGB}{250,80,220}
\begin{document}

\title{Margenau-Hill  distribution as a Necessary and Sufficient Signature of Measurement Incompatibility} 
\author{Partha Patra}
\email{parthapatra144@gmail.com}
\affiliation{Centre for Interdisciplinary Programs, Indian Institute of Technology Hyderabad, Telangana-502284, India}
\author{A. K. Pan}
	\email{akp@phy.iith.ac.in}
	 \affiliation{Department of Physics, Indian Institute of Technology Hyderabad, Telangana-502284, India}

\begin{abstract}
Measurement incompatibility and the negativity of quasi-probability distributions both demonstrate the signature of nonclassicality. However, their relations largely remained qualitative, and no explicit operational connection has been established. We provide an operational link between measurement incompatibility and the Margenau-Hill(MH) quasi-probability distribution associated with two dichotomic observables. {  We first derive the joint measurability (measurement compatibility) condition for any pair of dichotomic observables in arbitrary finite dimension $d$.} We then rigorously prove that for any pair of unsharp dichotomic measurements in dimension $d$, the positivity of the MH distribution is equivalent to joint measurability \emph{i.e.}, the MH distribution is positive \emph{if and only if} the measurements are jointly measurable. We further introduce the MH-like quasi-probability distribution for $n$ dichotomic unsharp observables in arbitrary finite dimension and derive a sufficient condition of joint-measurability when the observables are mutually anticommuting. Finally, we propose an interferometric setup, inspired by quantum-switch architectures, that directly reconstructs MH quasi-probability which in turn provides a test of incompatibility.
\end{abstract}
\maketitle

\emph{Introduction—}
The first noted departure of quantum theory from its classical counterpart is the  existence of measurements that cannot be performed simultaneously\cite{Born1925}. In other words, unlike classical physics, the joint probability distribution of some observables does not exist in quantum theory. For projective measurements, the noncommutativity of two observables successfully captures this notion. However, in quantum theory, a generalized notion of measurement exists in terms of positive-operator-valued measure(POVM) for which the concept of joint measurability or compatibility is introduced. 

The POVMs $\{E_{i}\}$ and $\{E_{j}\}$ associated with the measurements $M_{1}$ and $M_{2}$ respectively, are said to be jointly measurable or compatible \cite{busch1996quantum} if there exists a global POVM $\{G_{i,j}\}$ that satisfies the following properties. i) Positivity: $G_{i,j}\geq 0 \hspace*{0.6cm}\forall i,j$ ii) Completeness :$\sum_{i,j} G_{i,j}=1$ and iii) Reproducibility of marginals: $E_{i}=\sum_{j}G_{i,j}; \ E_{j}=\sum_{i}G_{i,j}$.  The existence of such a grand POVM validates that the measurements are compatible\cite{Uola2014}, otherwise the measurements are termed incompatible. The measurement incompatibility provides the basis for many fundamental nonclassical features such as steering\cite{quintino16,uola16,Quitino2015}, Bell nonlocality \cite{wolf09,Plavala2025} and  for a range of quantum information processing tasks \cite{Guhne2023,mukherjee24}.

On the other hand, the quasi-probability distribution of the outcomes of observables plays an important role in identifying a form of nonclassicality. Such a distribution satisfies all of Kolmogorov's axioms for probability,  but sometimes  can be negative, which is a signature of nonclassicality.  Historically, such a distribution for position and momentum was first proposed by Wigner \cite{Wigner1932}. Later, Kirkwood  \cite{KD1933} and Dirac \cite{Dirac1945} independently developed a quasi-probability distribution, which can be complex and the MH distribution \cite{MH1961,johansen2007} is the real part of it. The negativity of a quasi-probability distribution has been proven to be a resource in quantum computation\cite{Veitch_2012,Howard2014}, quantum information processing \cite{Ferrie_2011}, quantum tomography \cite{Lundeen2012,Thekkadath2016}, and quantum metrology \cite{Arvidsson-Shukur2020,Lupu-Gladstein2022}. Negative quasi-probability and contextuality have been shown to be equivalent notion of nonclassicality \cite{spekkens2008,Pusey2014}. Connections of the negative quasi-probability distribution with  the anomalous weak value \cite{Aharonov1988,Sasmal2025} and the violation of Leggett-Garg inequality \cite{Leggett1985,PAN2023indefinitecausalorder} have also been  made \cite{Dressel2015, AKP_Interference_2020}.

It is then a relevant question whether a connection can be established between these two notions of nonclassicality, viz., the negative quasi-probability and the  measurement incompatibility. Brief attempts have been made in this direction. In \cite{Rahimi-Keshari2021}, it is shown that the positivity of $S$-ordered phase space distribution of two POVMs is sufficient for their joint measurability. In \cite{ghai2023}, a qualitative connection has been made between the negativity of the generalized Wigner distribution and the degree of incompatibility.

In this work, we attempt to establish a direct connection between these two forms of nonclassicality. For a set of POVMs, the associated grand POVM is generally not unique, and identifying it becomes a challenging task in higher-dimensional systems \cite{busch1996quantum,PhysRevD_Busch1986,Banik2013}.  On the other hand, there remain many forms of quasi-probability distributions due to the possibility of ordering in various ways \cite{Glauber1963,Sudarshan1963,Husimi1940,Stephan2021,Hofer2017quasiprobability,Rahimi-Keshari2021}. We note that the negativity of quasi-probability is a signature of nonclassicality, but while it is positive, no conclusion can be made. This means that negative quasi-probability distributions constitute only a sufficient condition for non-classicality.  

 We note that the necessary and sufficient condition for joint measurability of two dichotomic observables is known only for $d=2$. We establish a proof of the necessary and sufficient condition for the joint measurability of any pair of dichotomic observables in arbitrary finite dimension $d$. We demonstrate that  the MH distribution constitutes the necessary and sufficient condition for the joint measurability of unsharp POVMs corresponding to these observables. Therefore, unlike other quasi-probability distributions, the positive MH distribution certifies the classicality, namely, the joint measurability. 

We extend our approach by suitably defining MH-like quasi-probability for more than two unsharp POVMs and derive the sufficient condition for joint measurability of POVMs corresponding to mutually anti-commuting observables in any arbitrary finite dimension. We further note that measurement incompatibility is, by itself, only a mathematical notion. Bridging it to quasi-probability opens up the opportunity to experimentally test it, as quasi-probability has been practically tested \cite{HernandezGomez2024Interferometry,Lostaglio2023kirkwooddirac,AKP_Interference_2020}. We propose an experimental scheme designed to directly measure the MH quasi-probability for two dichotomic observables, which in turn can be considered as the test of measurement incompatibility.

\emph{Joint measurability of arbitrary two dichotomic observables —}
{ We derive the necessary and sufficient condition of joint measurability of any two dichotomic observables in an arbitrary finite dimension $d$} by considering POVMs as the unsharp versions of projective measurements, given by
\begin{eqnarray}
\label{unpovm}
    E^{a_1}_{A_1}=\frac{\mathbb{I}+ \lambda a_1 A_1}{2}; \ \ \ E^{a_2}_{A_2}=\frac{\mathbb{I}+ \mu a_2 A_2}{2} 
\end{eqnarray}
where $\lambda ,\mu \in[0,1]$  are the unsharpness parameters and $a_1, a_2\in\{+1,-1\}$. Here, $A_{1}$ and $A_2$ are dichotomic observables in an arbitrary finite dimension and unsharp version of them are  $A_{1}^{\lambda}= E^{+}_{A_1}- E^{-}_{A_1}$ and $A_{2}^{\mu}= E^{+}_{A_2}- E^{-}_{A_2}$ respectively. 

The joint measurability condition for two observables in the qubit system is derived by Busch \cite{PhysRevD_Busch1986} as 
\begin{eqnarray}
\label{jmcond}
|| \lambda A_1 + \mu A_2||+||\lambda A_1 - \mu A_2||\leq 2
\end{eqnarray}
 {Note that Eq.(\ref{jmcond}) constitutes the necessary and sufficient condition for compatibility for $d=2$ \cite{JAE2019PRA,jae2026PRR}}. { If $\lambda=\mu$,  Eq. (\ref{jmcond}) directly gives that for $\lambda\leq \frac{1}{\sqrt{2}}$ any two qubit observables are jointly measurable. It is shown in \cite{Banik2013} that $\lambda\leq \frac{1}{\sqrt{2}}$ is the sufficient condition for joint measurability of any two observables in an arbitrary finite dimension $d$.} It is  proved in \cite{wolf09} that if a pair of dichotomic measurements is incompatible in quantum theory, then this incompatibility persists in any no-signaling theory. For the case $\lambda \neq \mu$, within any general probabilistic theory, any pair of dichotomic observables is jointly measurable if the condition $\lambda^2 + \mu^2 \leq 1$ holds \cite{Busch_2013}.

{ We again note here that Eq. (\ref{jmcond}) is valid only for qubit observables and its extension for the observables in higher dimension has not hitherto been provided. Here we present the following theorem that provides a necessary and sufficient condition for the joint measurability of two unsharp dichotomic observables in any arbitrary finite dimension $d$.

\begin{thm}
    Two dichotomic observables $A_1^{\lambda}$ and $A_2^{\mu}$ in any arbitrary finite dimension $d$ are jointly measurable if and only if 
\begin{eqnarray}
\label{JM_condition}
  \sqrt{\lambda^2  + \mu^2 + 2 \lambda \mu \operatorname{Spec}(C)}+\sqrt{\lambda^2  + \mu^2 -2  \lambda \mu \operatorname{Spec}(C)}\leq 2
\end{eqnarray}
where $C=\half\{A_1,A_2\}$ and $\operatorname{\operatorname{Spec}}(C)$ denotes the spectrum of eigenvalues of $C$.

\end{thm}

\begin{proof}
    We start by considering the most general form of grand POVM as,
    \begin{eqnarray}\label{GPOVM_2obr}
        G^{a_1,a_2}(Z)=\frac{1}{4}(\mathbb{I}+a_1\lambda A_1+ a_2 \mu A_2 + a_1a_2 \lambda \mu Z)
    \end{eqnarray}
    Here, $a_1,a_2\in\{\pm 1\}$ and $Z$ is a Hermitian operator. The marginals are reproduced as $E^{a_1}_{A_1}=\sum_{a_2}G^{a_1,a_2}(Z)$ and $E^{a_2}_{A_2}=\sum_{a_1}G^{a_1,a_2}(Z)$. Our goal is to find the condition of positivity of $G^{a_1,a_2}(Z)$.
    
    From Jordan's lemma \selectlanguage{french}\cite{Jordon_lemma}\selectlanguage{english}, it is known that the two dichotomic observables $A_1$ and $A_2$ acting on the $d$-dimensional Hilbert space $\mathcal{H}$ can be simultaneously block diagonalized into 1D and 2D subspaces, yielding a total of $\mathcal{R}$ such blocks. That is, the Hilbert space can be decomposed as $\mathcal{H}=\bigoplus_{r\in \mathcal{R}}\mathcal{H}_r$. Let ${\Pi}_r$ denote the orthogonal projector onto the subspace that simultaneously block diagonalizes $A_1$ and $A_2$, follows $\sum_r \Pi_r=\mathbb{I}_d$. 
We consider $\Tilde{Z}=\sum_r\Pi_r Z \Pi_r \equiv \sum_r \Tilde{Z}_r $ and present the following lemma.

\begin{Lemma}
    If $G^{a_1,a_2}(Z) \succeq 0$, then $G^{a_1,a_2}(\Tilde{Z}) \succeq 0$
\end{Lemma}
\begin{proof}
    Consider an arbitrary quantum state $\ket{\psi}$, which can be written as $\ket{\psi}=\sum_r \Pi_r\ket{\psi}\equiv \sum_r \ket{\psi_r}$. Then 
    \begin{eqnarray}
        \bra{\psi}G^{a_1,a_2}(\Tilde{Z})\ket{\psi}=&\bra{\psi}\sum_r \Pi_r G^{a_1,a_2}(Z)\Pi_r\ket{\psi}& \nonumber\\
        &=\sum_r\bra{\psi_r}G^{a_1,a_2}(Z)\ket{\psi_r}&
    \end{eqnarray}
    Since each term satisfies positivity $\bra{\psi_r}G^{a_1,a_2}(Z)\ket{\psi_r}\geq 0$, it follows that $\bra{\psi}G^{a_1,a_2}(\Tilde{Z})\ket{\psi}\geq0$.
\end{proof}

Consequently, without loss of any generality, the condition for joint measurability can be derived by considering the  grand POVM $G^{a_1,a_2}(\Tilde{Z})$.
 The positivity of $G^{a_1,a_2}(\Tilde{Z})$ demands the positivity of POVM in each subspace of the Jordan block. Explicitly,
 \begin{eqnarray}
     G^{a_1,a_2}_r(\Tilde{Z}_r)=\frac{1}{4}(\mathbb{I}_r+a_1\lambda A_{1_r}+ a_2 \mu A_{2_r} + a_1a_2 \lambda \mu \Tilde{Z}_r)\succeq 0
 \end{eqnarray}
 Therefore, for each 2D block, using  Eq~\eqref{jmcond}, the joint measurability condition can be written as
\begin{eqnarray}\label{JM_per_block1}
    \sqrt{\lambda^2+\mu^2+ \lambda \mu \{A_{1_r},A_{2_r}\}}+\sqrt{\lambda^2+\mu^2- \lambda \mu \{A_{1_r},A_{2_r}\}}\leq 2
\end{eqnarray}

Now, we define the quantity $c_r=\half \langle\{A_{1_r},A_{2_r}\}\rangle$. Since, anti-commutator of two qubit observables is proportional to $\mathbb{I}_2$, the quantity $c_r$ is a scalar in  each 2D block and is one of the eigenvalues of $\{A_{1_r},A_{2_r}\}$. Since $A_1$ and $A_2$ are block diagonal, therefore, $\{A_1,A_2\}$ is also block diagonal with eigenvalues $c_r$.

For the 1D block, $A_{1_r}$ and $A_{2_r}$ are scalars taking values $\pm 1$, and they are always jointly measurable irrespective of any value of unsharpness parameter $\lambda$ and $\mu$. Because their anticommutator satisfies $\{A_{1_r},A_{2_r}\} \in \{\pm 2\}$, they always satisfy the Eq.~\eqref{JM_per_block1} for each 1D block. Finally, from Eq. (\ref{JM_per_block1}) we obtain the joint measurability condition for each Jordan block as

\begin{eqnarray}\label{JM_per_block}
\sqrt{\lambda^2+\mu^2+2 \lambda \mu c_r}+\sqrt{\lambda^2+\mu^2-2 \lambda \mu c_r}\leq 2
\end{eqnarray}

It is evident that the general joint measurability condition is solely determined by the eigenvalues of $\{A_1,A_2\}$ and independent of any state. Consequently, any pair of two dichotomic observables in arbitrary finite dimension $d$ are jointly measurable if and only if they satisfy Eq.~\eqref{JM_condition}.
\end{proof}
}
\emph{The MH quasi-probability and joint measurability —} Given a density matrix $\rho$, the MH distribution \cite{MH1961} is defined as 
\begin{eqnarray}
q(a_1,a_2)=\frac{1}{2}Tr[(\pi^{a_1}_{A_1}\pi^{a_2}_{A_2}+\pi^{a_2}_{A_2}\pi^{a_1}_{A_1})\rho]\label{MHD}
\end{eqnarray}
satisfying the normalization condition   $\sum\limits_{a_{1},a_{2}} q(a_{1},a_{2})=1$.  Here $\pi^{a_{1}}_{A_1}$ and $\pi^{a_{2}}_{A_2}$ are projectors associated with dichotomic observables $A_1$ and $A_2$, respectively, where $A_i=\sum_{a_i=\pm 1}a_i \pi^{a_i}_{A_{i}}\ \forall i \in \{1,2\}$. Importantly, $q(a_1,a_2)$ provides the correct marginals as  

\begin{subequations}
\begin{eqnarray}
\label{oi1}
	p(a_2)=\sum\limits_{a_1}	q(a_1, a_2)= Tr[\pi^{a_2}_{A_2}\rho] \\
	\label{oi2}
	p(a_1)=\sum\limits_{a_2}	q(a_1, a_2) = Tr[\pi^{a_1}_{A_1}\rho] 
\end{eqnarray}
\end{subequations}
which may also be considered as the operational non-invasiveness.

Note that, the negativity of a given quasi-probability distribution ($q(a_1,a_2)< 0$) indicates nonclassical behavior \cite{Ferrie_2008,spekkens2008,Jeong_2020}, while its positivity does not necessarily imply classicality. In contrast, the existence of a global POVM, \emph{i.e.}, the compatibility of the two unsharp observables $A_{1}^{\lambda} \text{ and } A_{2}^{\mu}$ can also be taken as a hallmark of classical behavior. We demonstrate the necessary and sufficient criterion for nonclassicality by establishing a connection between these two manifestations of nonclassical features. For this, let us replace the projectors in the MH distribution in Eq.\eqref{MHD} by unsharp POVM elements, $i.e.$,
\begin{eqnarray}
\label{mhpovm}
    q(a_1,a_2)=\frac{1}{2}
    \bra{\psi}(E^{a_1}_{A_1}E^{a_2}_{A_2}+E^{a_2}_{A_2}E^{a_1}_{A_1})\ket{\psi}\label{qprb2Observable}
\end{eqnarray}
where $E^{a_1}_{A_1}$ and $E^{a_2}_{A_2}$ are defined in Eq. (\ref{unpovm}). Without loss of generality, we may take $\rho=|\psi\rangle\langle\psi|$, with $|\psi\rangle$ being an arbitrary state in any arbitrary finite dimension $d$. We prove the following theorem to demonstrate the connection between the positivity of the MH distribution in Eq.(\ref{mhpovm})  with the  joint measurability for any pair of dichotomic observables $A_1^{\lambda}$ and $A_2^{\mu}$ in any arbitrary finite dimension $d$.
\begin{thm}\label{thm1}
    For any two unsharp dichotomic observables $A_1^{\lambda}$ and $A_2^{\mu}$ in any arbitrary finite dimension $d$,\\
    
        Positivity of MH  distribution $\Leftrightarrow$ Joint measurability
\end{thm}

\begin{proof}
   Using Eq.~\eqref{qprb2Observable}, the positivity of MH distribution can explicitly be written as 
\begin{subequations}\label{Mhd_explicit}
\begin{eqnarray}
\label{qpp}
q(++)=\frac{1}{4}\langle\psi|(\mathbb{I}+\lambda A_1 +\mu A_2+ \lambda \mu C)|\psi\rangle\geq 0\\
\label{qpm}
q(+-)=\frac{1}{4}\langle\psi|(\mathbb{I}+\lambda A_1 -\mu A_2 - \lambda \mu C)|\psi\rangle \geq  0\\
\label{qmp}
q(-+)=\frac{1}{4}\langle\psi|(\mathbb{I}-\lambda A_1 +\mu A_2- \lambda \mu C)|\psi\rangle \geq  0\\
\label{qmm}
q(--)=\frac{1}{4}\langle\psi|(\mathbb{I}-\lambda A_1 -\mu A_2+ \lambda \mu C)|\psi\rangle\geq  0
\end{eqnarray}
\end{subequations}
where $C=\frac{\{A_1,A_2\}}{2}$ and $|\psi\rangle$ is arbitrary.

{ 
It therefore follows that, whenever the MH distribution is nonnegative, there exists a global POVM of the form given in Eq.~\eqref{GPOVM_2obr}, with $Z\equiv C=\half\{A_1,A_2\}$. In particular, one may define grand POVM as,
\begin{eqnarray}\label{MH_quasi_prb}
    G^{a_1,a_2}_{MH}=\frac{1}{4}\left(\mathbb{I}+ a_1\lambda A_1 + a_2\mu A_2+ a_1 a_2\lambda \mu C\right)\succeq 0,
\end{eqnarray}
where $a_1,a_2\in\pm 1$. The operators $G^{a_1,a_2}_{MH}$ are positive semidefinite and reproduce the required marginals. Consequently, the Eq.~\eqref{MH_quasi_prb} always satisfies the joint measurability criterion in Eq~\eqref{JM_condition}. This is due to the fact that Eq.~\eqref{JM_condition} is derived by considering general grand POVMs. Therefore, whenever the MH distribution is positive, the unsharp observables $A_1^{\lambda}$ and $A_2^{\mu}$ are jointly measurable.

To prove the converse, we need  to prove that if Eq.~(\ref{JM_condition}) holds then all four MH distribution are positive. From Eq.~(\ref{JM_condition}) for each block $\mathcal{H}_r$, we derive the following expression by re-arranging and squaring it, leading to two subsequent equations.
{\small\begin{subequations}
     \begin{eqnarray}
      \lambda^2+\mu^2 -2\lambda \mu  c_r \leq 4 + \left(\lambda^2+\mu^2 +2\lambda \mu  c_r\right)
 - 4\sqrt{\lambda^2+\mu^2 +2\lambda \mu  c_r}\nonumber\\ \\
 \lambda^2+\mu^2 +2\lambda \mu  c_r \leq 4 + \left(\lambda^2+\mu^2 -2\lambda \mu  c_r\right) 
 - 4\sqrt{\lambda^2+\mu^2 -2\lambda \mu  c_r}\nonumber\\
 \end{eqnarray}
\end{subequations}
}

Here, $c_r=\half \langle\{A_{1_r},A_{2_r}\}\rangle$ is the eigenvalue of $C$ in the $\mathcal{H}_r$ block. This leads to 
\begin{subequations}
     \begin{eqnarray}
 \sqrt{\lambda^2+\mu^2 +2\lambda \mu  c_r} \leq 1+\lambda \mu  c_r\\
 \sqrt{\lambda^2+\mu^2 -2\lambda \mu  c_r} \leq 1-\lambda \mu  c_r    
 \end{eqnarray}
\end{subequations}
respectively.  Note that $\sqrt{\lambda^2+\mu^2 \pm 2\lambda \mu  c_r}\equiv || (\lambda A_1 \pm \mu A_2)||_r, \ \forall {r\in \mathcal{R}}$. Again, Cauchy–Schwarz inequality implies,
     $|\langle \psi_r | \lambda A_{1_r} \pm \mu A_{2_r}|\psi_r\rangle| \leq || \lambda A_{1_r} \pm \mu A_{2_r}|| \forall{r\in \mathcal{R}}$ and we obtain,
 \begin{eqnarray}
     |\bra{\psi_r} \lambda A_{1_r} \pm \mu A_{2_r}\ket{\psi_r}| \leq 1 \pm \lambda \mu \bra{\psi_r} c_r \mathbb{I}_r\ket{\psi_r}
 \end{eqnarray}
 
 Here, $C=\bigoplus_r c_r \mathbb{I}_r\equiv \half\{A_1,A_2\}$ is block diagonal in the same decomposition as $A_1$ and $A_2$. Thus, we can write
 {\small\begin{eqnarray}
     \sum_{r\in \mathcal{R}}|\left(\langle \psi| \Pi_r\right) ( \lambda A_{1} \pm \mu A_{2}) (\Pi_r|\psi\rangle)| \leq \sum_{r\in \mathcal{R}}\left(\langle \psi| \Pi_r\right)( \mathbb{I} \pm \lambda \mu C) (\Pi_r|\psi\rangle) 
     \end{eqnarray}}
 
 Using triangle inequality, $\left|\sum_ru_r\right|\leq\sum_r|u_r|$ we obtain,
 \begin{eqnarray}
\label{lastEq}
  \left|\langle \psi|(\lambda A_1 \pm \mu A_2)|\psi\rangle\right|\leq 1\pm\lambda \mu \langle \psi|C|\psi\rangle
\end{eqnarray}
 
 That directly indicates that all four MH quasi-probabilities Eqs.~(\ref{qpp}-\ref{qmm}) are positive. This completes the proof.}
\end{proof}

Since the positivity of MH distribution implies joint measurability and converse also holds, we conclude that MH distribution for two dichotomic unsharp observables constitutes a necessary and sufficient condition for classicality. However, the similar conclusion may not be made for other forms of quasi-probability distribution as positivity of them remains inconclusive.   

\emph{Generalization to an arbitrary number of unbiased POVMs —}
We extend the above formalism by defining a MH-like quasi-probability distribution for an arbitrary number of dichotomic unsharp observables. For three unbiased dichotomic POVMs, we define a quasi-probability distribution as

\begin{eqnarray}
    q(a_1,a_2,a_3)=\frac{1}{4}\bra{\psi}\{E^{a_1}_{A_1},\{E^{a_2}_{A_2},E^{a_3}_{A_3}\}\}\ket{\psi}\label{q123}
\end{eqnarray}
Here, we consider the unsharpness parameter $\lambda_i$ for every observable $A_i$ where the POVM element is $E^{a_i}_{A_i}=\frac{\mathbb{I}+\lambda_i a_i A_i}{2}$, where $i\in\{1,2,3\}$. From Eq.~\eqref{q123}, we then get
\begin{eqnarray}
    &q(a_1,a_2,a_3)
  = \frac{1}{8}\bra{\psi}\Big[\mathbb{I}
    +  \big(\lambda_1 a_1 A_1+ \lambda_2  a_2 A_2+\lambda_3  a_3 A_3\big)& \nonumber\\
  &+ \frac{1}{2} \big(\lambda_1 \lambda_2 a_1a_2\{A_1,A_2\}  + \lambda_2 \lambda_3 a_2a_3 \{A_2,A_3\}+ \lambda_3 \lambda_1 a_3a_1 \{A_3,A_1\}\big) &\nonumber\\
  &+ \frac{\lambda_1 \lambda_2 \lambda_3}{4} a_1a_2a_3\{A_1, \{A_2,A_3\}\}\Big]\ket{\psi}&
  \label{3ObservableGPOVM}
\end{eqnarray}

It can be readily verified that  $q(a_1,a_2,a_3)$ satisfies the condition of normalization $\sum_{a_1,a_2,a_3} q(a_1,a_2,a_3)=1$, and reproduce the marginals, 
\begin{eqnarray}
 p(a_k) \equiv \sum_{a_i,a_j\in\{+1,-1\}} q(a_i,a_j,a_k) \ \ \forall i,j,k\in\{1,2,3\}
\end{eqnarray} 
The positivity of all eight quasi-probabilities $q(a_1,a_2,a_3)$ imposes a criterion for the joint measurability of three dichotomic unsharp observables $A_1^{\lambda_1},A_2^{\lambda_2}$ and $A_3^{\lambda_3}$ of arbitrary dimension as,  
\begin{eqnarray}
    & \sum_{a_1,a_2,a_3\in\{+1,-1\}}|| \lambda_1 a_1 A_1+ \lambda_2 a_2 A_2+ \lambda_3 a_3 A_3& \nonumber \\ &+\frac{\lambda_1 \lambda_2 \lambda_3}{4} a_1a_2a_3\{A_1, \{A_2,A_3\}\}||\leq 8& \label{jtc3Ob}
\end{eqnarray}
This is a sufficient condition for joint measurability of three unsharp POVMs. However, for the case of mutually anticommuting observables $\{A_i\}_{i=1}^{3}$, Eq.~\eqref{jtc3Ob} becomes 
\begin{eqnarray}
    \sum_{a_1,a_2,a_3\in\{+1,-1\}}||\lambda_1 a_1 A_1+ \lambda_2 a_2 A_2+ \lambda_3 a_3 A_3||\leq 8 \label{JM_3obsr}
\end{eqnarray}

Now, if one considers the same unsharpness parameter $\lambda$ for all three anti-commuting observables, then Eq. (\ref{JM_3obsr}) gives the necessary and sufficient conditions. This directly implies that $\lambda \leq \frac{1}{\sqrt{3}}$. For three anti-commuting qubit observables, this is the well-known  condition for joint measurability \cite{CarmeliIW2019,LIANG20111}. Our construction shows that precisely the same bound governs the joint measurability of three dichotomic observables in arbitrary finite dimension. Hence, whenever Eq.~\eqref{JM_3obsr} is satisfied, the three dichotomic observables $A_1^{\lambda_1}, A_2^{\lambda_2} \text{ and } A_3^{\lambda_3}$ are jointly measurable. In this way, Eq.~\eqref{JM_3obsr} provides a simple and dimension-independent sufficient condition for the joint measurability of three dichotomic observables. Note that we derived the Eq.~\eqref{JM_3obsr} from the positivity of the quasi-probability which has previously been identified as sufficient condition for joint measurability condition in  general probabilistic theory framework \cite{mukherjee24}.  

For an arbitrary $n$ number of unbiased dichotomic observables, we propose the MH-like quasi-probability as,
\begin{eqnarray}
q(a_1,a_2,\dots ,a_n)=\frac{1}{2^{n-1}}\bra{\psi}\{E^{a_1}_{A_1},\dots,\{E^{a_{n-1}}_{A_{n-1}},E^{a_n}_{A_{n}}\}\}\ket{\psi}\label{GPOVM_N_Observ}
\end{eqnarray}
  By considering $q(a_1,a_2,\dots ,a_n)>0$, from Eq. (\ref{GPOVM_N_Observ}), we can derive the sufficient condition for joint measurability similar to Eq. (\ref{3ObservableGPOVM}). For
  $n$ mutually anticommuting observables in an arbitrary finite dimension we derive the sufficient condition for joint measurability as

\begin{eqnarray}
\label{njm}
\sum_{a_1,a_2,\dots ,a_n\in\{+1,-1\}}||\sum_{i\in\{1,2,\dots,n\}} \lambda_i a_iA_i||\leq 2^{n}
\end{eqnarray}
It is proved \cite{mukherjee24} that for the case $\lambda_i=\lambda\  \forall i$, the above relation  is the sufficient condition for joint measurability in general probabilistic theories. 
Utilizing the fact that the observables are anticommuting, it is straightforward to find, for $\lambda_i=\lambda$ $\forall i$, we have $\lambda \leq \frac{1}{\sqrt{n}}$, for which $n$ dichotomic unbiased measurements are compatible. Therefore, positivity of MH-like quasi-probability distribution $q(a_1,a_2,\dots ,a_n)$ serves as a sufficient condition for a notion of classicality, namely, the  joint measurability of $n$ dichotomic observables in any arbitrary finite dimension.

\emph{Direct test of joint measurability—}
We start by noting that the negative quasi-probability has been tested experimentally \cite{Lostaglio2023kirkwooddirac,Ryu2019,HernandezGomez2024Interferometry}. The Theorem~\ref{thm1} establishes that  the experimental test of MH distribution directly test the measurement incompatibility. To test the MH distribution we introduce an experimental scheme motivated from the quantum switch experiment \cite{Chiribella2013,Goswami2018,guo2020,BAN2021,PAN2023indefinitecausalorder}.


We consider a modified Mach–Zehnder interferometer, shown in Fig.~\ref{Fig1}. We use two different degrees of freedom of a single photon, the path $\{\ket{\psi_1},\ket{\psi_2}\}$ and the polarization $\{\ket{H},\ket{V}\}$. Consider the initial state of the photon is $\ket{\psi}\ket{\chi}$. Upon entering the first beam splitter (BS1), the photon splits into two paths, directed along the vertical path $\ket{\psi_1}$ and the horizontal path $\ket{\psi_2}$ so that the state of the photon becomes
\begin{eqnarray}
   |\Psi\rangle= \frac{1}{\sqrt{2}} \big(\ket{\psi_1} + \ket{\psi_2} \big)\ket{\chi},
\end{eqnarray}
where $\ket{\chi}=\alpha\ket{H}+\beta \ket{V}$ represents the state of the system (polarization).

\begin{figure}[ht]
\includegraphics[height=65mm, width=85mm]{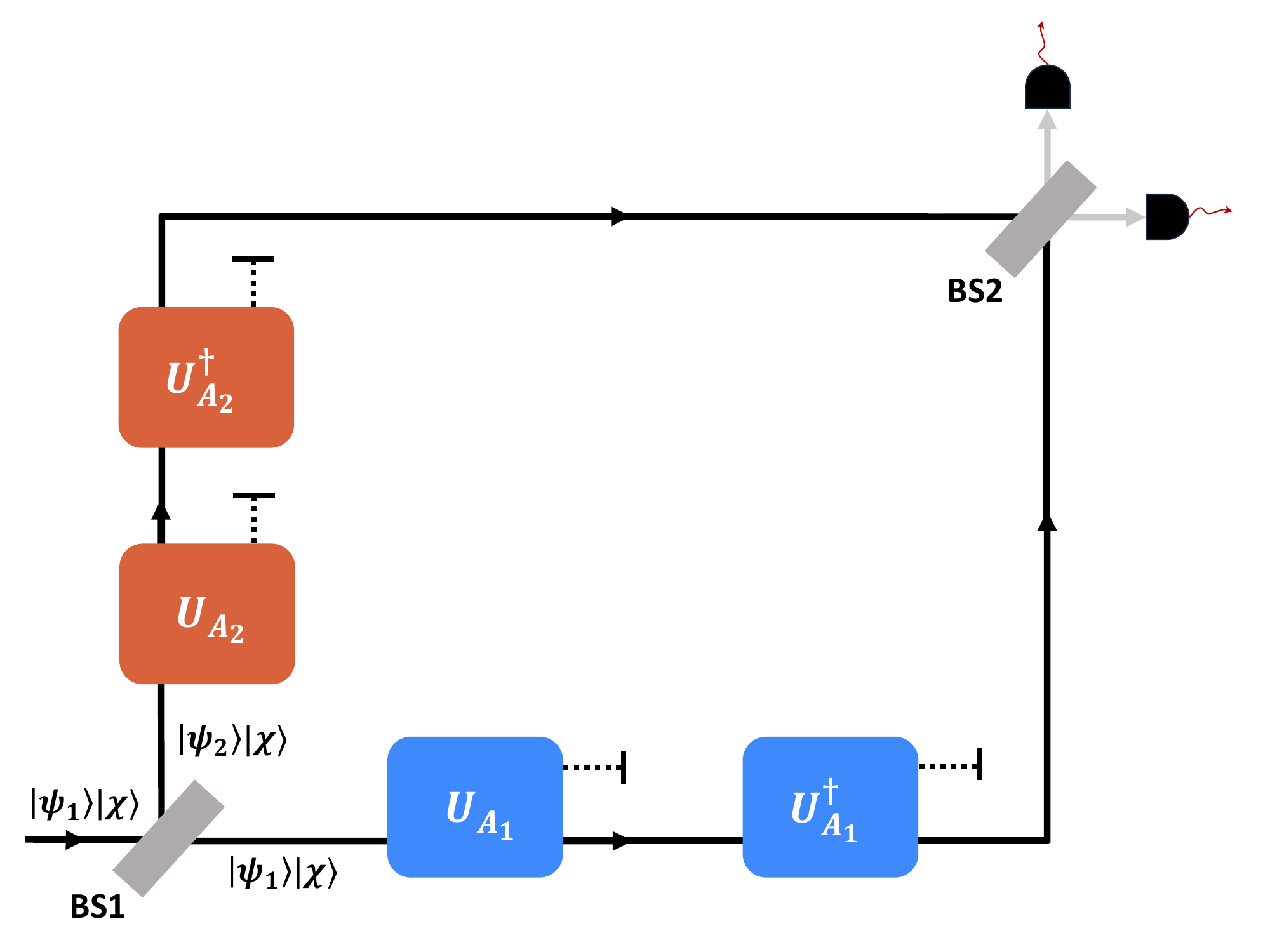}
\caption{An interferometric setup for testing MH quasi-probability distribution. See text for detailed description.}\label{Fig1}
\end{figure}

 For our purpose, we employ the well-known stinespring dilation \cite{Stinespring1955} to the polarization degree of freedom  by introducing an ancilla, so that the measurement outcome is stored in the ancilla without destroying coherence. The action of stinespring dilation through the unitary $U_A$ corresponds to the measurement $A \equiv \{N_A^a\}_a$ is given by,
\begin{eqnarray}
   \ket{\chi}\ket{0}\xrightarrow{U_A}\sum_{a\in \{\pm1\}} N^a_A\ket{\chi}\ket{a}
\end{eqnarray}
where $N^{a}_A$ is the Kraus operator with $(N^{a}_A)^{\dagger}N^{a}_A+(N^{\bar{a}}_A)^{\dagger}N^{\bar{a}}_A=\mathbb{I}$.

In our setup, along the paths  $\ket{\psi_1}$ and $\ket{\psi_2}$, we introduce two different unitaries, $U_{A_1}$ and $U_{A_2}$ respectively, associated with the observables $A_1$ and $A_2$ . Here,  $E_{A_1}^{a_1}$ and $E_{A_2}^{a_2}$ are POVM elements corresponding to outcomes $a_1$ and $a_2$ , and $E_{A_i}^{a_i} = (N_{A_i}^{a_i})^{\dagger} N_{A_i}^{a_i}$ for $i \in \{1,2\}$, where $N_{A_i}^{a_i}$ denotes the Kraus operators. We take the following form of unitary,
\begin{eqnarray}
    U_{A_i} =  N^{a_i}_{A_i} \otimes |0\rangle\langle 0| + \sqrt{\mathbb{I} - (N^{a_i}_{A_i})^{\dagger}N^{a_i}_{A_i}} \otimes |1\rangle\langle 0|\nonumber\\
    + \sqrt{\mathbb{I} -N^{a_i}_{A_i}(N^{a_i}_{A_i})^{\dagger}} \otimes |0\rangle\langle 1| - (N^{a_i}_{A_i})^{\dagger}\otimes |1\rangle\langle 1|
\end{eqnarray}
Along the path $\ket{\psi_1}$ ($\ket{\psi_2}$), the photon first passes through unitary operation $U_{A_1}$ ($U_{A_2}$), after which the ancilla is postselected in the state $\ket{0}$. Subsequently, the photon is subjected to the inverse operation $U^{\dagger}_{A_1}$ ($U^{\dagger}_{A_2}$), again followed by postselection of the ancilla in the state $\ket{0}$. Note here that the unitary and its inverse are applied in sequence. It may seem that the state remains unchanged but because the ancillary Hilbert space is different in the two cases, this does not happen. Therefore, the state before BS2 becomes,
\begin{eqnarray}
    \ket{\Psi^{\prime}}&=\frac{\mathcal{C}}{\sqrt{2}} \big(\ket{\psi_1} \otimes (N^{ a_1}_{A_1})^{\dagger}N^{a_1}_{A_1}\ket{\chi}+ \ket{\psi_2} \otimes (N^{ a_2}_{A_2})^{\dagger} N^{a_2}_{A_2}\ket{\chi} \big)&\nonumber\\
   & \equiv \frac{\mathcal{C}}{\sqrt{2}} \big(\ket{\psi_1} \otimes E^{a_1}_{A_1}\ket{\chi}+ \ket{\psi_2} \otimes E^{a_2}_{A_2}\ket{\chi} \big)&
\end{eqnarray}
Due to postselection we have to re-normalize the state and  $\mathcal{C}=\sqrt{\frac{2}{\bra{\chi}( E_{A_1}^{a_1})^2+( E_{A_2}^{a_2})^2\ket{\chi}}}$ , is the normalization constant. Upon the action of BS$2$, the states transform as $\ket{\psi_1} \xrightarrow{BS2} \frac{1}{\sqrt{2}}\big(\ket{\psi_3} + \ket{\psi_4}\big)$ and $\ket{\psi_2} \xrightarrow{BS2} \frac{1}{\sqrt{2}}\big(\ket{\psi_3} - \ket{\psi_4}\big)$, yielding the final state,
\begin{eqnarray}
     \ket{\Psi_F}=\frac{\mathcal{C}}{{2}}  \left( \ket{\psi_3}\otimes \bigl(E^{a_1}_{A_1}+E^{a_2}_{A_2}\bigr) \ket{\chi}+ \ket{\psi_4}\otimes\bigl(E^{a_1}_{A_1}-E^{a_2}_{A_2}\bigr) \ket{\chi}\right) \nonumber\\
\end{eqnarray}

Now, we consider the measurement of the path observable $P_{\Psi}=\ketbra{\psi_3}{\psi_3}- \ketbra{\psi_4}{\psi_4}$, whose expectation  value corresponds to the difference of photon counts between two paths $\ket{\psi_3}$ and $\ket{\psi_4}$. The expectation value of $P_{\Psi}$ is then given by, 
\begin{eqnarray}
    \langle P_{\Psi}\rangle=\bra{\Psi_F}\Big(\ketbra{\psi_3}{\psi_3}\otimes \mathbb{I}- \ketbra{\psi_4}{\psi_4}\otimes \mathbb{I}\Big)\ket{\Psi_F} 
\end{eqnarray}
which yields the MH quasi-probability up to a positive scaling factor, as

\begin{eqnarray}
    \langle P_{\Psi}\rangle=\frac{|\mathcal{C}|^2}{4}\bra{\chi}\Bigl(E^{a_1}_{A_1}E^{a_2}_{A_2}+E^{a_2}_{A_2}E^{a_1}_{A_1}\Bigr)\ket{\chi}\equiv\frac{|\mathcal{C}|^2}{2} q(a_1,a_2),
\end{eqnarray}
Hence, the interferometric setup in Fig. \ref{Fig1} gives the MH quasi-probability distribution  through the expectation value of $P_{\Psi}$. This in turn provides a test of  measurement incompatibility. In particular, any observed negativity of $q(a_1,a_2)$ experimentally certifies the incompatibility of two unsharp observables $A_{1}^{\lambda} \text{ and } A_{2}^{\mu}$.

\emph{Summary and Outlook—}
In sum,  our work establishes a direct connection between  the nonclassical behaviors featured in negative  quasi-probability and measurement incompatibility of unsharp dichotomic POVM in any arbitrary finite dimension $d$. {We derive a necessary and sufficient condition for joint measurability of any pair of unsharp dichotomic POVM in any arbitrary finite dimension $d$ in Eq.~\eqref{JM_condition}.} We further demonstrate an operational equivalence between the positive MH distribution and the joint measurability, $i.e.,$ the positivity of this distribution validate the compatibility of two arbitrary dichotomic measurements and converse also holds true. We further propose a MH-like quasi-probability distribution for an arbitrary $n$ number of unbiased dichotomic observables in any arbitrary finite dimension and  derive a sufficient condition of joint measurability for $n$ mutually anticommuting observables. Finally, we put forward an interferometric arrangement, inspired by quantum-switch architectures, that directly measures the MH quasi-probability and thus offers an experimental method for testing measurement incompatibility.

One promising future direction is to generalize the present scenario beyond dichotomic unbiased measurement to scenarios with more than two outcomes and for the case of biased POVMs. For the $n$ unsharp POVMs, strengthening the currently known sufficient condition to a necessary and sufficient one, and more precisely identifying the corresponding quasi-probability distribution, would constitute another interesting avenue for future research. 

\acknowledgments
 P.P. acknowledges funding from the University Grants Commission (NTA Ref. No.-231610049800), Govt. of India. A.K.P. acknowledges the support from the Research Grant SERB/CRG/2021/004258, Government of India. 

\bibliography{ref}
\end{document}